\documentclass[aps,pra,twocolumn,showpacs,superscriptaddress]{revtex4-2}
													
\usepackage{amsmath}
\usepackage{amsthm}
\usepackage{latexsym}
\usepackage{amssymb}
\usepackage{bm}
\usepackage{bbm}
\usepackage{appendix}
\usepackage{graphics,epstopdf}
\usepackage{color}
\usepackage[colorlinks=true,linkcolor=blue,citecolor=red,plainpages=false,pdfpagelabels]{hyperref}
\usepackage[normalem]{ulem}

\usepackage{newlfont}
\usepackage{amsfonts}
\usepackage{amsthm}
\usepackage{graphicx}
\usepackage{epsfig}

\usepackage[thinc]{esdiff}

\usepackage{times}
\newtheorem{lemma}{Lemma}

\def\tr{\operatorname{tr}}

\newcommand{\bra}[1]{\langle #1|}
\newcommand{\ket}[1]{|#1\rangle}

\newcommand{\op}[1]{|#1\rangle\!\langle#1|}
\newcommand{\norm}[1]{\Vert #1\Vert}
\newcommand{\abs}[1]{\vert #1\vert}
\newtheorem{theorem}{Theorem}

\DeclareOldFontCommand{\rm}{\normalfont\rmfamily}{\mathrm}

\begin{document}

\title{Quantum speed limits for information and coherence}

\author{Brij Mohan}\email{brijmohan@hri.res.in}

\affiliation{Harish-Chandra Research Institute,\\  A CI of Homi Bhabha National
Institute, Chhatnag Road, Jhunsi, Prayagraj 211019, India
}

\author{Siddhartha Das}\email{das.seed@gmail.com}
\affiliation{Centre for Quantum Information \& Communication (QuIC),  \'Ecole polytechnique de Bruxelles, Universit\'e libre de Bruxelles, Brussels, B-1050, Belgium}
\affiliation{Center for Security, Theory and Algorithmic Research (CSTAR), Centre for Quantum Science and Technology (CQST), International Institute of Information Technology, Hyderabad, Gachibowli, Telangana 500032, India}

\author{Arun Kumar Pati}\email{akpati@hri.res.in}
\affiliation{Harish-Chandra Research Institute,\\  A CI of Homi Bhabha National
Institute, Chhatnag Road, Jhunsi, Prayagraj 211019, India
}

\begin{abstract}
The quantum speed limit indicates the maximal evolution speed of the quantum system. In this work, we determine speed limits on the informational measures, namely the von Neumann entropy, maximal information, and coherence of quantum systems evolving under dynamical processes. These speed limits ascertain the fundamental limitations on the evolution time required by the quantum systems for the changes in their informational measures. Erasing of quantum information to reset the memory for future use is crucial for quantum computing devices. We use the speed limit on the maximal information to obtain the minimum time required to erase the information of quantum systems via some quantum processes of interest.
\end{abstract}

\maketitle
\section{Introduction}
In quantum information theory, the speed limit aims to address the question, ``What is the maximal rate of change of the state of a quantum system in a physical process". It provides the minimum time needed for a quantum system to evolve from a given initial state to the target state. Therefore, it sets the fundamental lower bound on the evolution time of the quantum system~\cite{Mandelstam1945, Margolus1998, Levitin2009, Anandan1990}. Determination of quantum speed limit (QSL) is important in the areas of quantum metrology~\cite{Campbell2018}, quantum control~\cite{Campbell2017}, quantum thermodynamics~\cite{Mukhopadhyay2018,Funo2019}, etc. For instances, it plays a crucial role in determining the minimum charging time of quantum batteries~\cite{F.Campaioli2018} and finding the minimum time required to implement quantum gates in quantum computation~\cite{AGN12}. Recently, the reverse quantum speed limit has been derived using the geometry of the quantum state space,
and its application in the quantum battery was also discussed~\cite{Mohan2021}. The study of quantum speed limits on informational measures is relevant for both fundamental and applied aspects, given the rapid progress in the area of quantum technologies and need to control the dynamics of quantum systems. 

Quantum superposition is one of the key features of the quantum theory and the amount of superposition present in a particular state is measured by quantum coherence. Quantum coherence is a crucial resource for several quantum information processing tasks~\cite{Giovannetti2011}, including quantum computation~\cite{Jeong2002,Ralph2003,Yoo2014,Srivastava2021} and quantum thermodynamics~\cite{Kammer2016, Singh2021}, where the energy eigenbasis forms a natural choice for the reference basis (e.g.,~\cite{Singh2021}). Quantification of resources such as the entanglement and the coherence are well-established approaches. It has been shown that the incoherent operations can convert any degree of coherence with respect to a reference basis into entanglement~\cite{Streltsov2015}. The notion of coherence naturally arises in the context of QSL, as shown in the Refs.~\cite{Marvian2016, D.Mondal2016,Rossatto2020}. So a natural question that arises is ``how fast the quantum coherence of the given quantum state can be generated or destroyed in a physical process". We have studied fundamental limits on the time required for the change in information content and quantum coherence which answers these questions. One of the consequences of our main results is that it sets speed limit for the quantum erasure operations, where the quantum states are reset to a fixed state.

 QSL was first obtained by Mandelstam and Tamm, which is based on the quantum fluctuations in energy~\cite{Mandelstam1945}. Later, another bound was obtained by Margolus and Levitin based on the mean energy~\cite{Margolus1998}. The maximum of these two bound is tight as shown in Ref.~\cite{Levitin2009}. If the quantum system evolves along the shortest geodesic, then the QSL is saturated and the kind of Hamiltonian that may drive the system along the geodesic was discussed in Ref.~\cite{Pati1993}. QSL has been widely studied for unitary dynamics~\cite{Mandelstam1945,Margolus1998,Levitin2009,Anandan1990,Gislason1956,Eberly1973,Bauer1978,Bhattacharyya1983,Leubner1985,Vaidman1992,Uhlmann1992,Uffink1993,Pfeifer1995,Horesh1998,Soderholm1999,Giovannetti2004,Andrecut2004,Gray2005,Luo2005,Batle2005,Borras2006,Zielinski2006,Zander2007,Andrews2007,Kupferman2008,Yurtsever2010,Fu2010,Jones2010,Chau2010,Zwierz2012,S.Deffner2013,Fung2013,Poggi2013,Fung2014,Andersson2014,D.Mondal2016,Mondal2016,S.Deffner2017,Campaioli2018}, open system dynamics~\cite{Deffner2013,Campo2013,Taddei2013,Pires2016,Jing2016,Pintos2021}, and arbitrary dynamics~\cite{S.Deffner2020}. Recently, resource speed limit~\cite{Camaioli2020} has been introduced, which describes how quickly quantum resources can be generated or degraded by physical processes. In  Ref.~\cite{Rudnicki20201}, QSL of quantum entanglement has been studied, using a geometric measure of entanglement. Therefore, the study of QSL for quantum state as well as various resources associated with a quantum system is of prime importance.

In this paper, we derive speed limits for the von Neumann entropy, maximal information, and quantum coherence for quantum systems 
undergoing arbitrary dynamical processes. These speed limits depict fundamental limitations on the minimal time required for the changes in the entropy, maximal information, and quantum coherence of quantum systems undergoing dynamical processes. Here, a dynamical process refers to a completely positive trace preserving (CPTP) map acting on a quantum system. Our result suggests a new bound on the rate of information production~\cite{S.Deffner2020}. As an application, we show that the speed limit on the maximal information provides the minimum time required to erase information of a quantum system in the context of Landauer's erasure~\cite{Landauer,Bennett,Goold2016}. Similarly, the speed 
limit on the information can also be applied for the thermalization process. We discuss speed limits on information and coherence and illustrate their applications by presenting some examples of dynamical processes of interest. Thus, we believe that our results will have
applications in quantum computing, quantum communication, quantum control, and quantum thermodynamics.

Our work is organised as follows. In Section~\ref{sec:prem}, we discuss the preliminaries and background required to arrive at the main results of this paper, which are presented in Section~\ref{sec:QSL}. In Section~\ref{sec:qsl-ic}, we obtain the speed limit on entropy and information. We discuss limitations on the minimal time required for erasing processes based on the speed limits on the information. In Section~\ref{sec:coherence-change}, we derive the speed limit on (basis dependent) quantum coherence. In Section~\ref{AQSLI}, we also discuss the speed limit bounds based on instantaneous evolution speed for informational measures. Finally, in the last section, we provide the conclusions.

 \section{Preliminaries}\label{sec:prem}
 Let $\cal{H}$ denote the separable Hilbert space, where $\dim(\mathcal{H})$ may be finite or infinite. The state of a quantum system is represented by a density operator $\rho$, which satisfies following properties: $\rho=\rho^{\dagger}$, $\rho\geq 0$, and $\tr[\rho]=1$. The identity operator is denoted by $\mathbbm{1}$. The physical transformation of the state of a system is given by a completely positive, trace-preserving map, which is also called quantum channel. Time-evolution of a quantum system evolving under a given dynamical process is given by the master equation 
 \begin{equation}
      \dot{\rho_t}:=\diff{\rho_t}{t}=\mathcal{L}_t(\rho_t),
 \end{equation}
 where $\rho_t$ is the state of the system at time $t=t$ and $\mathcal{L}_t$ is the Liouvillian super-operator~\cite{Rivas2012}, which in general can be time dependent or time independent.
 
 The von Neumann entropy $S(\rho)$ of a quantum system in the state $\rho$ is defined as \begin{equation}\label{eq:entropy}
  S(\rho):= -\tr\{\rho \ln \rho\}.\end{equation}
The entropy $S(\rho)$ can be interpreted as the average information content of a quantum system in the state $\rho$, which can be defined on finite- or infinite-dimensional Hilbert space. The maximal information $I(\rho)$ of a finite-dimensional quantum system with the Hilbert space of dimension $d$ and in the state $\rho$ is defined as~\cite{Zurek1983}:
\begin{equation}\label{eq:information}
    I(\rho):= \ln(d) - S(\rho).
\end{equation}
We call $I(\rho)$ as the maximal information as it can be interpreted as the maximum amount of information that can be gained by performing optimal measurements on the quantum system~\cite{Zurek1983}. It also provides a measure of objective information for a quantum state $\rho$~\cite{Open2003}.
  
  We now recall Theorem 1 of Ref.~\cite{Das2018} below which we use in this work. 
  \begin{lemma}[\cite{Das2018}]\label{lem:ent-change}
 For any quantum dynamical process with $\dim(\mathcal{H})<+\infty$, the rate of entropy change is given by 
 \begin{equation}
   \frac{d}{dt} S(\rho_t)= -{\rm tr}\{\dot{\rho_t}  \ln \rho_t\},
\end{equation}
  whenever $\dot{\rho}_{t}$ is well-defined. The above formula also holds when $\dim(\mathcal{H})=+\infty$ given that $\dot{\rho}_{t}\ln \rho_{t}$ is trace-class and the sum of the time derivative of the eigenvalues of $\rho_{t}$ is uniformly convergent
on some neighborhood of $t$, however small.
  \end{lemma}
   
 The operator norm $\norm{A}_{\rm op}$, the Hilbert-Schmidt norm $\norm{A}_{\rm HS}$, and the trace norm $\norm{A}_{\rm tr}$ of an operator $A$ are defined as $\norm{A}_{\rm op}:=\lambda_{\max}$ where $\lambda_{\max}$ is the maximum of the absolute value of eigenvalues of $A$ when $A$ is Hermitian, $\norm{A}_{\rm HS}:=\sqrt{{\rm tr}(A^{\dagger}A)}$, and $\norm{A}_{\rm tr}={{\rm tr}(\sqrt{A^{\dagger}A}})$, respectively.

We recall standard definitions relevant in the context of coherence from Ref.~\cite{Baumgratz2014}, where a resource theoretical framework for quantifying the quantum coherence was proposed. Let us consider a quantum system associated with a finite-dimensional Hilbert space $\cal{H}^d$ such that $\dim(\cal{H}^d)=d$. Let $\{|i\rangle\}_{i=0}^{d-1}$ be the reference eigenbasis of $\mathcal{H}^d$. The state of the quantum system described by density operator $\rho$ $\in$ $\cal{D}(\cal{H}^d)$, where $\cal{D}(\cal{H}^d)$ is convex set of density operators. The subset  of incoherent states $\cal{I}$ $\subset$ $\cal{D}(\cal{H}^d)$ consist of the family of diagonal density matrices $\omega$ = $\sum_{i}p_{i}\op{i}$ (with  $0  \leq p_{i} \leq 1$ and $\sum_{i}p_{i}=1$) in the reference eigenbasis $\{|i\rangle\}_{i=0}^{d-1}$. In brief, any reasonable measure of the quantum coherence $C(\rho)$ must fulfill the following conditions as introduced in Refs.~\cite{Baumgratz2014,Streltsov2017}: ({\it i}) must be  non-negative real number for all state $\rho$, with $C(\rho)$ = 0 iff $\rho$ $\in $ $\cal{I}$; ({\it ii}) do not increase under the mixing of quantum states (convexity) $C(\sum_{k}p_{k}\rho_{k})$ $\le$ $\sum_{k}p_{k}C(\rho_{k})$, with $\rho_{k}$ $\in$ $\cal{D(H)}$, 0  $\le$ $p_{k}$  $\le$ 1 and $\sum_{k}p_{k}$ = 1; ({\it iii}) must be monotonic under incoherent completely positive and trace-preserving (ICPTP) map (i.e., incoherent quantum channel), $C(\Phi_{ICPTP}(\rho))$ $\leq$ $C(\rho)$ for all $\Phi_{ICPTP}$ maps; ({\it iv}) must be monotonic under selective measurements i.e. $C(\rho)$ $\ge$ $\sum_{k}p_{k}C(\rho_{k})$, where $\rho_{k}=E_{k}\rho E_{k}^{\dagger}/p_{k}$ (with $p_{k}$ = ${\rm tr}(E_{k}\rho E_{k}^{\dagger})$) are the states after measurement for arbitrary Kraus operators $\{E_{k}\}$, which obey the condition $\sum_{k}E_{k}^{\dagger}E_{k}$ = $\mathbbm{1}$ and  $E_{k}\cal{I}E_{k}^{\dagger}$ $\subset$ $\cal{I}$.

There are several widely known (basis dependent) quantum coherence measures such as the relative entropy of coherence~\cite{Baumgratz2014}, the $l_{1}$ norm of coherence~\cite{Baumgratz2014}, the geometric coherence~\cite{Streltsov2015}, and the robustness of coherence~\cite{Napoli2016}, etc.
We are using the relative entropy of coherence because of its operational meaning as the distillable coherence~\cite{winter2016,Singh2017}. In addition, it is also easier to work and compute compared to some other measures of coherence. For a given state $\rho$, the relative entropy of coherence defined as
\begin{equation}\label{eq:coherence}
     C(\rho):= S(\rho^{\rm D})-S(\rho),
\end{equation}
where $\rho^{\rm D}:=\sum_{i}\langle i|\rho|i \rangle \op{i}$ is the density operator that is diagonal in the reference basis, obtained by dephasing off-diagonal elements of $\rho$. The reference basis is fixed and independent of time. The von Neumann entropy $S(\rho^{\rm D})$ for the diagonal density operator $\rho^{\rm D}$ reduces to the Shannon entropy $H(\{p_i\}):=-\sum_ip_i\ln p_i$, where the probability distribution $\{p_i\}$ is given by $p_i:= \bra{i} \!\rho\!\ket{i}$. We call a quantum channel, i.e., a CPTP map, to be completely dephasing with respect to the reference basis $\{\ket{i}\}$ if it acts on a density operator $\rho$ and yields $\rho^{\rm D}=\sum_{i}\langle i|\rho|i \rangle \op{i}$. 

Apart from the basis dependent notion of coherence measure~\eqref{eq:coherence}, in Ref.~\cite{Ma2019}, the maximal information $I(\rho)$~\eqref{eq:information} is argued to be basis independent coherence measure with reference to only incoherent state $\frac{\mathbbm{1}}{d}$.

\section{Quantum Speed Limit}\label{sec:QSL}
In general, quantum speed limits illustrate fundamental limitations on the evolution of quantum systems due to given quantum dynamics. In this section, we present the speed limits on the entropy $S(\rho)$, maximal information $I(\rho)$, and quantum coherence $C(\rho)$. 

Entropy is a widely studied fundamental quantity in quantum information theory. Entropy of a state can be interpreted as an average (expected) information content of the given quantum state. The joint evolution of the system and environment is considered to be a unitary operation in quantum theory. It is known that the unitary process keeps entropy invariant. However, the local evolution of the system alone can be non-unitary, i.e., some noisy physical process. This non-unitarity process causes a flow of information between the system and the environment, which can change the entropy of the system. The rate of the entropy change is also related to the rate of the change for some entanglement measures of the system (e.g., see Section III of Ref.~\cite{Das2018}). These aspects of the entropy motivate us to define lower bounds on the minimal time require for the entropy change, i.e., speed limits on the entropy.

 Quantum coherence is a fundamental non-classical property of quantum systems, which act as a resource for several quantum processing tasks (e.g.,~\cite{Giovannetti2011,Singh2021}). So the natural question arises ``how fast can the quantum coherence in the given quantum state be generated, destroyed or erased by some physical process". To answer this question, the resource speed limit was introduced in Ref.~\cite{Camaioli2020},  which is also applicable for quantum coherence. There the resource speed limit is defined using the divergence-based measures. The bound obtained in Ref.~\cite{Camaioli2020} is challenging to calculate in general as it requires optimization over all free states. Here, we have obtained QSL on coherence using the relative entropy of coherence, which is arguably easier to calculate. We also derive QSL on the maximal information $I(\rho)$ which can also be interpreted as a basis independent measure of coherence~\cite{Ma2019} and a measure of objective information of given quantum state~\cite{Open2003}.
 
  Our method of obtaining speed limits is similar to the technique used in Refs.~\cite{Deffner2013,Cai2017}. We briefly discuss concerns and other method based on Ref.~\cite{Connor2021} to derive other speed limits in Section~\ref{AQSLI}.

\subsection{Quantum speed limit for information}\label{sec:qsl-ic}
We now discuss the first main result of this work that provides a lower bound on evolution time of entropy, $T \geq T_{ESL}$ using the 
Hilbert Schmidt norm. 
The second main theorem of this work provides a lower bound on evolution time of the entropy or equivalently for the information, $T\geq T_{ISL}$ using the operator norm.

\begin{theorem}\label{thm:1}
 For an arbitrary quantum dynamics describable as time-evolution, the minimum time needed for the state $\rho_{t}$ to attain entropy $S(\rho_{T})$, where $\rho_{T}:=\rho_{t=T}$, starting with the initial entropy $S(\rho_{0})$, where $\rho_{0}:=\rho_{t=0}$, is lower bounded by
\begin{equation}\label{eq:theorem-1}
T \geq  T_{ESL}=  \frac{\abs{S(\rho_T) -S(\rho_0)}}{ {\Lambda}^{\rm rms}_{T}\overline{\norm{{ \ln \rho_t}}^{2}_{\rm HS}}},
\end{equation}
where ${\Lambda}^{\rm rms}_T:=\sqrt{\frac{1}{T}\int_{0}^{T}  \norm{\mathcal{L}_{t}({\rho_t})}^{2}_{\rm HS}{\rm d}t}$ is the root mean square evolution speed of the quantum system, $\overline{\norm{{\ln \rho_t}}^{2}_{\rm HS}}:=\sqrt{\frac{1}{T}\int_{0}^{T}  \norm{{ \ln \rho_t}}^{2}_{\rm HS} {\rm d}t}$, and $\cal{L}_{t}$ is the Liouvillian super-operator.
This theorem holds for both finite-dimensional and infinite-dimensional systems~\cite{Das2018}.
\end{theorem}

\begin{proof}
 The entropy of time evolved state $\rho_{t}$ given by
\begin{equation}
   S(\rho_t)= -{\rm tr}\{\rho_t  \ln \rho_t\}.
\end{equation}
After differentiating the above equation with respect to time $t$, we obtain~\cite{Das2018}
\begin{equation}
   \frac{{\rm d}}{{\rm d}t} S(\rho_t)= -{\rm tr}\{\dot{\rho_t}  \ln \rho_t\}=-{\rm tr}\{\mathcal{L}_{t}({\rho_t})  \ln \rho_t\}.
\end{equation}
Let us now consider the absolute value of the above equation and apply the Cauchy--Schwarz inequality $|{\rm tr}(AB)|\leq\sqrt{{\rm tr}(A^{\dagger}A){\rm tr}(B^{\dagger}B)}$. We then obtain the following inequality
\begin{equation}\label{eq:ent-bound}
   \left|\frac{{\rm d} }{{\rm d}t}S(\rho_t)\right|= |{\rm tr}\{ \mathcal{L}_{t}({\rho_t}) \ln \rho_t\}| \leq \norm{\mathcal{L}_{t}({\rho_t})}_{\rm HS} \norm{ \ln \rho_t}_{\rm HS}.
\end{equation}
The above inequality~\eqref{eq:ent-bound} is the upper bound on that the rate of change of the entropy of the quantum system evolving under given dynamics. 
After integrating above equation with respect to time $t$, we obtain 
\begin{equation}\label{MER}
  \int_{0}^{T} {\rm d}t\left|\frac{{\rm d} }{{\rm d}t}S(\rho_t)\right| \leq \int_{0}^{T} \norm{\mathcal{L}_{t}({\rho_t})}_{\rm HS} \norm{ \ln \rho_t}_{\rm HS}   {\rm d}t.
\end{equation}
Now applying the Cauchy--Schwarz inequality on the right hand side of the above inequality, we get  
\begin{equation}
  \int_{0}^{T} {\rm d}t\left|\frac{{\rm d} }{{\rm d}t}S(\rho_t)\right| \leq \sqrt{\int_{0}^{T}  \norm{\mathcal{L}_{t}({\rho_t})}^{2}_{\rm HS}{\rm d}t} \sqrt{\int_{0}^{T}  \norm{ \ln \rho_t}^{2}_{\rm HS} {\rm d}t}.
\end{equation}
From the above inequality, we get the desired bound:
\begin{equation}
  T\geq  \frac{|S(\rho_T) -S(\rho_0)|}{ {\Lambda}^{\rm rms}_{T}\overline{\norm{{ \ln \rho_t}}^{2}_{\rm HS}}}.
\end{equation}
\end{proof}
The lower bound of $T_{ESL}$ in \eqref{eq:theorem-1} is positive if the entropy of the system changes under the dynamical process and is zero if there is no change in the entropy. The change in entropy is zero when either the dynamics of the quantum system is unitary~\cite{Das2018,Nayak2007,Rivas2012} or the initial state is the fixed point of the governing dynamics. The $T_{ESL}$ is zero when the state is not evolving under the dynamics, i.e., when it is a fixed point of the governing dynamics.

The minimal time $T_{ESL}$ would be zero when entropy change is zero even if the state undergoes non-trivial transformation under dynamics, i.e., there is a change in the state during the process.

We now derive the second main result of this work that provides a lower bound on evolution time of information or basis-independent coherence $T\geq T_{ISL}$.

\begin{theorem}\label{thm:2}
 For an arbitrary quantum dynamics of a finite-dimensional system describable as time-evolution of its state, the minimum time needed for the state $\rho_{t}$ to attain information $I(\rho_{T})$, where $\rho_{T}:=\rho_{t=T}$, starting with the initial information $I(\rho_{0})$, where $\rho_{0}:=\rho_{t=0}$, is lower bounded by
\begin{equation}\label{eq:info-t}
T \geq  T_{ISL}=  \frac{\abs{I(\rho_T) - I(\rho_0)}}{ {\Lambda}^{\rm rms}_{T}\overline{\norm{{\ln \rho_t}}^{2}_{\rm op}}}, 
\end{equation}
where ${\Lambda}^{\rm rms}_T:=\sqrt{\frac{1}{T}\int_{0}^{T}  \norm{\cal{L}_{t}({\rho_t})}^{2}_{\rm tr}{\rm d}t}$ is the root  mean  square  evolution  speed  of the  quantum  system,  $\overline{\norm{ \ln \rho_t}^{2}_{\rm op}}:=\sqrt{\frac{1}{T}\int_{0}^{T}  \norm{{\ln \rho_t}}^{2}_{\rm op} {\rm d}t}$, and $\cal{L}_{t}$ is the Liouvillian super-operator.
\end{theorem}
\begin{proof}
We adapt the proof arguments of Theorem~\ref{thm:1} to arrive at the above bound. Detailed proof is provided in Appendix~\ref{QSLISL}, which somewhat differs from the proof of Theorem~\ref{thm:1}.

\end{proof}
Note that $\abs{I(\rho_T) - I(\rho_0)}=\abs{S(\rho_0)-S(\rho_T)}$ for any finite-dimensional system evolving under the dynamics that preserves its dimension. Therefore, the inequality~\eqref{eq:info-t} also provides speed limit on entropy change using the operator norm. The entropy does not change when either the dynamics is unitary~\cite{Das2018,Nayak2007,Rivas2012} or the initial state is the fixed point of the governing dynamics. Then the change in maximal information is zero, and $T_{ISL}$ is also zero. However, the minimal evolution time for state evolution is non-zero in the unitary process.

The above two results will have important applications in quantum computing and quantum communication where 
 change in the entropy is inevitable. Our results show that quantum theory imposes limits on the rate of change of the entropy and information. The maximal rate of information production has interesting applications ranging from Black hole to quantum communication. For example, the Bekenstein bound~\cite{Bakenstein1973,Bakenstein1981} says the rate with which information can be retrieved from black
 hole. Our results show that the rate of information (entropy) gain~\cite{Holevo2010,Buscemi2016} in a quantum system obeys nontrivial bound for general physical processes.
 
We will show that the maximum rate of the information change satisfies a new bound.
The inequality~\eqref{eqs} can be expressed as
\begin{align}
\mathcal{\dot{I}}=\overline{\left|\frac{{\rm d} I(\rho_t)}{{\rm d}t}\right|}  \leq {{\Lambda}^{rms}_T  \overline{\norm{ \ln \rho_t}_{\rm op}^{2}}},
\end{align}
 where $\overline{\left|\frac{{\rm d} I(\rho_t)}{{\rm d}t}\right|}:= \frac{1}{T} \int_{0}^{T}{\rm d}t \left|\frac{{\rm d} I(\rho_t)}{{\rm d}t}\right|$.
 The above equation provides the upper bound on that the rate of change of the information change of the quantum system for arbitrary quantum dynamics~\cite{S.Deffner2020}.
 
 {\it Minimum time for Landauer's erasure}.---The process of resetting an input bit to a fixed bit value is one of the most elementary operation in classical as well as quantum computing. Resetting a bit has a thermodynamic cost. 
 The Landauer erasure principle says that to erase a single qubit we must spent $k_{B}T\ln2$ amount of energy or equivalently,
 we have to spent some amount of entropy~\cite{Landauer,Bennett,Goold2016}.
 In recent years, there have been a great surge of interest in trying to improve and generalise the Landauer principle for information erasure~\cite{Reeb2014,Goold2015,Lorenzo2015,Scandi2019,Abiuso2020,Timpanaro2020}.  Notably, a universal bound on the energy cost of resetting operation has been proved for finite time~\cite{Miller2020,Zhen2021}. Once we know that the bit reset happens in finite time, it is natural to ask is there any non-trivial bound on the speed of erasing an input bit. 
 The results proved in Theorem~\ref{thm:1} and Theorem~\ref{thm:2} can answer an important question: How fast one can erase information in a physical system? 
 
 Suppose, we have an arbitrary qubit initially prepared in the state $\ket{\psi}$. The erasure operation will transform  $\ket{\psi} \rightarrow \ket{\Sigma}$, where $\ket{\Sigma}$  is a fixed state (often called a blank state). This
transformation cannot be realised by reversible operation and
hence involves some energy cost as put forth by Landauer. Now,
let us consider a general state of the system at time t = 0
which is given by $\rho_0$ and the input state has information $I(\rho_0) = \ln d - S(\rho_0)$, where $\rho_0$ acts on the Hilbert space
of dimension $d$. Under the action of the erasure operation,
which is a CPTP map (quantum channel), any quantum state
gets mapped to a fixed state i.e., $\rho_0 \rightarrow \op{\Sigma}$. For example, the fixed state can be 
the ground state of a two-level system which is a pure state with the
lowest energy. Since the final state is a fixed
pure state and the effective dimension of the final state is one
(as the fixed state lives in a one-dimensional subspace), we
have $I(\rho_T) = 0$. Therefore, the bound~\eqref{eq:info-t} can be used to
state that the minimum time required to erase information is
given by
 \begin{equation}
     T_{Erasure}= \frac{I(\rho_{0})}{{\Lambda}^{\rm rms}_{T}\overline{\norm{{\ln \rho_t}}^{2}_{\rm op}}}.
 \end{equation}
 Thus, while Landauer's erasure principle says how much minimum energy one needs to spend to erase a single bit or qubit, our result answers
 the question on what is the minimum time needed to erase a single bit or qubit. In the example below, we use the above bound to answer the question how fast erasure happens in thermalisation process.
 
 {\it Erasing information via thermalization}.--- Let us consider a two-level atom with ground state $\op{1}$  and excited state $\op{0}$, which is weakly interacting with a heat (thermal) bath at a fixed temperature. This interaction of the atom with the heat bath would transform the state of the atom to a thermal state. Depending on the Hamiltonian of the atom, the larger the gap between two energy levels, closer the thermal state is with the ground state. This is the idea behind erasing of the information in the system via thermalization~\cite{Lubkin1987,Plenio2001}. The jump operators of the heat bath for a two-level system are given as $L_{-} = \sqrt{\gamma_{0}(N+1)}\sigma_{-} $ and  $L_{+} = \sqrt{\gamma_{0}N}\sigma_{+}$, where $\sigma_{-}= |1\rangle\langle0|$ and  $\sigma_{+}=|0\rangle\langle1|$ are the lowering and raising operators to the system, $N$ is the mean number of photons in the resonant bath, $\gamma_{0}$ denotes the spontaneous emission rate of the bath, and $\gamma=\gamma_{0}(2N+1)$ is the total emission rate. The Lindblad master equation~\cite{Lindbland1976} governs the time evolution of atom and it is given by
\begin{align}
\frac{\partial \rho_{t}}{\partial t}& = {\gamma_{0}(N+1)}(\sigma_{-}\rho_t\sigma_{+}-\frac{1}{2}\{\sigma_{+}\sigma_{-},\rho_t\} \nonumber\\
 &+{\gamma_{0}N}(\sigma_{+}\rho_t\sigma_{-}-\frac{1}{2}\{\sigma_{-}\sigma_{+},\rho_t\}).
\end{align}
 If the atom the initially in a state  $\rho_{0}=\op{\psi(0)}$, where $|\psi(0)\rangle = {\cos\frac{\theta}{2}}|0\rangle +  \sin{\frac{\theta}{2}}|1\rangle$, then solution of the Lindbland equation is given by~\cite{Cherian2019}
\begin{align}\label{rho:ther}
\rho_t=&\frac{1}{2} \left(1-\frac{\gamma_{0}}{\gamma }+e^{-\gamma  t} \left(\frac{\gamma_{0}}{\gamma }+\cos \theta \right)\right) |0\rangle \langle 0 | \nonumber\\
&+\frac{1}{2} \left(1+\frac{\gamma_{0}}{\gamma }-e^{-\gamma  t} \left(\frac{\gamma_{0}}{\gamma }+\cos \theta \right)\right)|1\rangle \langle 1|\nonumber\\
&+ \frac{1}{2}e^{-\gamma t} \sin{\theta}  (|1\rangle \langle 0|+ |0\rangle \langle 1|).
\end{align}
To estimate bound~\eqref{eq:info-t}, we require the following quantities:
 \begin{equation}
     I(\rho_{0})= \ln{2},
 \end{equation}
\begin{align}
 I(\rho_{t})=& \ln{2} 
   +\frac{ \left(\gamma  -e^{-\gamma t}\sqrt{\delta}\right)}{2 \gamma }
   \ln{[\frac{ \left(\gamma -e^{-\gamma t}\sqrt{\delta}\right)}{2 \gamma }]}\nonumber\\
    &+\frac{ \left(\gamma  +e^{-\gamma t}\sqrt{\delta}\right)}{2 \gamma }
   \ln{[\frac{\left(\gamma +e^{-\gamma t}\sqrt{\delta}\right)}{2 \gamma }]},
\end{align}
\begin{align}
&\norm{ \cal{L}(\rho)}_{\tr}=  \frac{1}{2}\gamma_{0} e^{- \gamma  t}\times\nonumber\\
 &\sqrt{ (2 N+1)^2 \sin ^2\theta +\frac{4 \left(\gamma_{0}+\gamma  (2 N+1) \cos \theta +2 \gamma_{0} N\right)^2}{\gamma ^2}},
 \end{align}
  \begin{equation}
\norm{\ln{\rho_{t}}}_{\rm op}= \max{\{|\lambda_{1}|,|\lambda_{2}|\}},
 \end{equation}
where $\delta=\gamma ^2+\gamma_{0}^2(1-2 e^{\gamma  t}+ e^{2 \gamma  t})+2 \gamma  \gamma_{0} (1- e^{\gamma  t} )\cos \theta$. $\lambda_{1}$ and $\lambda_{2}$ are eigenvalues of $\ln{\rho_{t}}$.
\begin{figure}[htp]
    \centering
    \includegraphics[width=8.5cm]{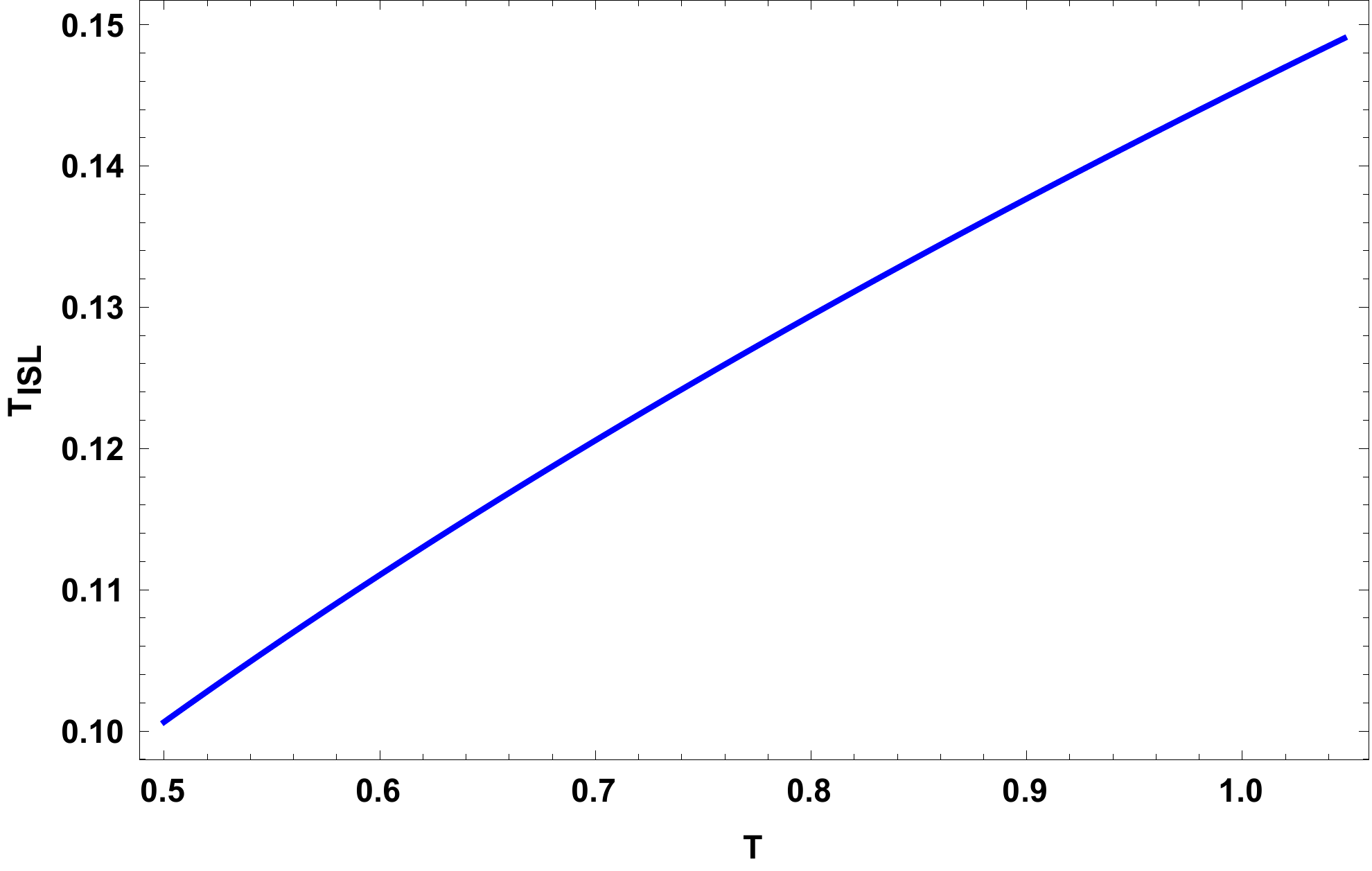}
    \caption{Here we depict $T_{ISL}$ vs $T$ and we have considered $\gamma_{0}=1$, $N=100 $, $\gamma=201$, and $\theta= \frac{\pi}{3}$. The process of erasure takes finite time and the example considered here does not saturate the bound. }
    \label{fig:TISL}
\end{figure}

In Fig.~\ref{fig:TISL}, we plot $T_{ISL}$ vs $T$ $\in$ $[0.5,\frac{\pi}{3}]$ for thermalization process and we have considered $\gamma_{0}=1$, $N=100 $, $\gamma=201$ and $\theta= \frac{\pi}{3}$. Inspecting eq.~\eqref{rho:ther} we see that the state $\rho_t$ requires an infinite amount of time to thermalize. That is, the thermalization process takes an infinite amount of time to erase information from the quantum system. Here in Fig~\ref{fig:TISL} we plot the bound~\eqref{eq:info-t} for finite time duration and the bound~\eqref{eq:info-t} is not tight. In Appendix~\ref{app:sat}, we give a proof that for open system dynamics, in general the bound cannot be reached.

With the recent advances in quantum technology, it may be possible to test the speed limit for the erasure operation. This may play an important role in intermediate-scale-noisy quantum computer where we do not have many physical qubits and we would like to reset our memory much faster in order to reuse them for another task.

\subsection{Quantum speed limit for coherence}\label{sec:coherence-change}
 We now derive the third main result of this work that provides a lower bound on evolution time of basis-dependent coherence $T\geq T_{CSL}$, where the reference basis $\{\ket{i}\}_{i=0}^{d-1}$ is fixed, independent of time.

\begin{theorem}\label{thm:3}
  For an arbitrary quantum dynamics of a finite-dimensional quantum system describable as time-evolution of its state, the minimum time needed for the state $\rho_{t}$ to attain coherence $C(\rho_{T})$, where $\rho_{T}:=\rho_{t=T}$, starting with the initial coherence $C(\rho_{0})$, where $\rho_{0}:=\rho_{t=0}$, is lower bounded by
\begin{equation}\label{csl}
 T \geq T_{CSL} = \frac{\abs{C(\rho_T) - C(\rho_0)}}{ {{\Lambda}}^{\rm rms,D}_{T}\overline{\norm{\ln \rho^{\rm D}_t}^{2}_{\rm HS}} +  {\Lambda}^{\rm rms}_{T}\overline{\norm{{\ln \rho_t}}^{2}_{\rm HS}}  },
\end{equation}
where ${\Lambda}^{\rm rms, \rm D}_T:=\sqrt{\frac{1}{T}\int_{0}^{T}  \norm{\mathcal{L}_{t}({\rho^{\rm D}_t})}^{2}_{\rm HS}{\rm d}t}$ is the root mean square evolution speed of the classical part of quantum system, ${\Lambda}^{\rm rms}_T:=\sqrt{\frac{1}{T}\int_{0}^{T}  \norm{\mathcal{L}_{t}({\rho_t})}^{2}_{\rm HS}{\rm d}t}$ is the root mean square evolution speed of the quantum system, $\overline{\norm{\ln \rho_t^{\rm D}}^{2}_{\rm HS}}:=\sqrt{\frac{1}{T}\int_{0}^{T}  \norm{{\ln \rho_t^{\rm D}}}^{2}_{\rm HS} {\rm d}t}$, $\overline{\norm{{\ln \rho_t}}^{2}_{\rm HS}}:=\sqrt{\frac{1}{T}\int_{0}^{T}  \norm{{\ln \rho_t}}^{2}_{\rm HS}{\rm d} t}$, and $\cal{L}_{t}$ is the Liouvillian super-operator.
\end{theorem}

\begin{proof}

The relative entropy of coherence of the time evolved quantum state $\rho_{t}$ is given as
\begin{equation}
     C(\rho_t)= S(\rho^{\rm D}_{t})-S(\rho_{t}).
\end{equation}
After differentiating above equation with respect to
time $t$ (i.e., employing [Lemma~\ref{lem:co-change}, Appendix~\ref{app:co-change}]), we obtain 
\begin{equation}
   \frac{{\rm d} }{{\rm d}t}C(\rho_t)= -{\rm tr}\{\mathcal{L}_{t}({\rho^{\rm D}_t}) \ln \rho^{\rm D}_t\}+{\rm tr}\{\mathcal{L}_{t}({\rho_t}) \ln \rho_t\}.
\end{equation}
Taking the absolute value of the terms in the above equation and applying the triangular inequality, we get
\begin{equation}
   \left|\frac{{\rm d}}{{\rm d}t} C(\rho_t)\right|\leq |{\rm tr}\{\mathcal{L}_{t}({\rho^{\rm D}_t}) \ln \rho^{\rm D}_t)\}|+|{\rm tr}\{\mathcal{L}_{t}({\rho_t}) \ln \rho_t\}|. 
\end{equation}
Then applying the Cauchy--Schwarz inequality, we obtain the following inequality
\begin{equation}
   \left|\frac{{\rm d}}{{\rm d}t} C(\rho_t)\right|\leq \norm{\mathcal{L}_{t}({\rho^{\rm D}_t})}_{\rm HS} \norm{\ln \rho^{\rm D}_t}_{\rm HS} +\norm{\mathcal{L}_{t}({\rho_t})}_{\rm HS} \norm{\ln \rho_t}_{\rm HS}.
\end{equation}
The above inequality represents the upper bound on that the rate of change of the basis-dependent coherence of the quantum system for arbitrary quantum dynamics. After integrating above equation with respect to time $t$, we obtain 
\begin{align}\label{MCR}
  \int_{0}^{T} {\rm d}t\left|\frac{{\rm d} }{{\rm d}t}C(\rho_t)\right| &\leq  \int_{0}^{T}  \norm{\mathcal{L}_{t}({\rho^{\rm D}_t})}_{\rm HS} \norm{\ln \rho^{\rm D}_t}_{\rm HS}  {\rm d}t\nonumber \\
  &+\int_{0}^{T}  \norm{\mathcal{L}_{t}({\rho_t})}_{\rm HS} \norm{{\ln \rho_t}}_{\rm HS} {\rm d}t.
\end{align}
Let us apply the Cauchy--Schwarz inequality on the right hand side of the above inequality, we get 

\begin{align}
  \int_{0}^{T}{\rm d}t\left|\frac{{\rm d}}{{\rm d}t} C(\rho_t)\right| \leq& \sqrt{\int_{0}^{T}  \norm{\mathcal{L}_{t}({\rho^{\rm D}_t})}^{2}_{\rm HS}{\rm d}t} \sqrt{\int_{0}^{T}  \norm{\ln \rho^{\rm D}_t}^{2}_{\rm HS} {\rm d}t}\nonumber\\
  &+\sqrt{\int_{0}^{T}  \norm{\mathcal{L}_{t}({\rho_t})}^{2}_{\rm HS}{\rm d}t} \sqrt{\int_{0}^{T}  \norm{{\ln \rho_t}}^{2}_{\rm HS} {\rm d}t}.
\end{align}
The above inequality can be written as
\begin{equation}\label{ch}
  T\geq  \frac{|C(\rho_T) - C(\rho_0)|}{ {{\Lambda}}^{\rm rms, \rm D}_{T}\overline{\norm{\ln \rho^{\rm D}_t}^{2}_{\rm HS}} +  {\Lambda}^{\rm rms}_{T}\overline{\norm{{\ln \rho_t}}^{2}_{\rm HS}}  }.
\end{equation}
\end{proof}
The quantum speed limit on coherence applies to both coherence generation and coherence degradation processes. In particular, our bound in \eqref{csl} can answer how fast a system undergoes decoherence. For any completely dephasing process, the above bound~\eqref{csl} is interpreted as the speed limit of decoherence, which means the minimum time required for a coherent state to become an incoherent state. If during the evolution diagonal part of $\rho_t$ in the reference basis is static then ${\Lambda}^{\rm rms, \rm D}_T$ becomes zero.

 The change in coherence is zero when either the dynamics of the quantum system are coherence-preserving or the initial state is a fixed point of the dynamics. The minimal time $T_{CSL}$ is zero when the change in coherence is zero even if there is a change in the state during the dynamics.

Next, we apply our bound on the speed of coherence in a couple of quantum dynamics of interest, namely the pure dephasing process and dissipative process (e.g.,~\cite{lidar2019}).

{\it Pure dephasing process.}--- Let us consider a two-level atom with the ground state $\op{1}$  and the excited state $\op{0}$ interacting with a dephasing environment. The corresponding dephasing jump operator is given by $L = \sqrt{\frac{\gamma}{2}} \sigma_{z}$, where  $\sigma_{z}$ is the Pauli-$Z$ operator and $\gamma$ is a real parameter denoting the strength of dephasing. The Lindblad master equation~\cite{Lindbland1976} governs the time evolution of atom, and it is given by
\begin{equation}
\frac{{\rm d} \rho_t}{{\rm d} t} = \cal{L}(\rho_t)= \frac{\gamma}{2}(\sigma_{z}\rho_t\sigma_{z}-\rho_t).
\end{equation}
If the atom the initially in a state  $\rho_{0}=\op{\psi(0)}$, where $|\psi(0)\rangle = {\cos\frac{\theta}{2}}|0\rangle +  \sin{\frac{\theta}{2}}|1\rangle$, then solution of the Lindbland equation is given by
\begin{align}\label{rho:dep}
\rho_t=& \cos^2 {\frac{\theta}{2}}\op{0} + \sin^2 {\frac{\theta}{2}}\op{1} \nonumber\\
&+ e^{-\gamma t} \sin {\frac{\theta}{2}} \cos {\frac{\theta}{2}} (|1\rangle\! \langle 0|+ |0\rangle \!\langle 1|).
\end{align}
To estimate bound \eqref{csl}, we require the following quantities:
\begin{align}
&|C(\rho_{t})-C(\rho_{0})|=\frac{1}{2} [\ln \left(e^{-\gamma t} \sqrt{\sin ^2\theta+ e^{2 \gamma  t}\cos ^2\theta }+1\right)\nonumber\\
&\qquad +\ln \left(\frac{1}{4}-\frac{1}{4} e^{-\gamma t} \sqrt{\sin ^2\theta + e^{2 \gamma  t}\cos ^2\theta }\right) \nonumber\\
&\qquad +2 e^{-\gamma t} \sqrt{\sin ^2\theta + e^{2 \gamma  t}\cos ^2\theta }\times\nonumber\\
&\qquad \coth ^{-1}\left(\frac{e^{\gamma  t}}{\sqrt{\sin ^2\theta +e^{2 \gamma  t}\cos ^2\theta  }}\right)],
\end{align}
\begin{align}
\norm{\cal{L}(\rho_{t})}_{HS}^{2}=\frac{1}{2} \gamma ^2 e^{-2 \gamma  t}\sin ^2\theta,
\end{align}
\begin{align}
\norm{\cal{L}(\rho_{t}^{\rm D})}_{HS}^{2}=0, 
\end{align}
\begin{align}
\norm{\ln{\rho^{\rm D}_{t}}}_{HS}^{2}=\left[\ln \left(\sin ^2\frac{\theta }{2}\right)\right]^2+\left[\ln \left(\cos ^2\frac{\theta }{2}\right)\right]^2,
\end{align}
\begin{align}
&\norm{\ln{\rho_{t}}}_{HS}^{2}=\nonumber\\
&\qquad\left[\ln \frac{ \left(2 + \sqrt{2} e^{-\gamma  t} \sqrt{ \left(e^{2 \gamma  t}-1\right)\cos 2\theta + e^{2 \gamma  t}+1}\right)}{4}\right]^2\nonumber\\
&+ \left[\ln \frac{ \left(2 - \sqrt{2} e^{-\gamma  t} \sqrt{ \left(e^{2 \gamma  t}-1\right)\cos 2\theta +e^{2 \gamma  t}+1}\right)}{4}\right]^2.
\end{align}
Here, we have calculated coherence in the computational basis $\{\ket{0},\ket{1}\}$. The diagonal component of $\rho_t$ in the reference basis is static during the pure dephasing process, therefore ${\Lambda}^{\rm rms, \rm D}_T$ is zero. In Fig.~\ref{fig:TCS1}, we plot $T_{CSL}$ vs $T$ $\in$ $[0, \frac{\pi}{3}]$ for pure dephasing dynamics and we have considered $\gamma=2$ and $\theta$ $\in$ $\{\frac{\pi}{2},\frac{\pi}{3},\frac{\pi}{4}\}$.  Note that, for the pure dephasing process, the speed limit of basis-dependent coherence can be interpreted as the speed limit of decoherence.

 In Fig.~\ref{fig:TCS1}, we can see that in the dephasing process, the maximally coherent state has a higher speed limit time of decoherence compared to other states. According to Eq. \eqref{rho:dep} the dephasing time for the state $\rho_t$ is infinite. That is, the dephasing process takes an infinite amount of time to erase coherence from the quantum system. In Fig.~\ref{fig:TCS1}, we plot the bound~\eqref{csl} for finite time duration and the bound~\eqref{csl} is not tight in general (see Appendix~\ref{app:sat}).

\begin{figure}[htp]
    \centering
    \includegraphics[width=8.5cm]{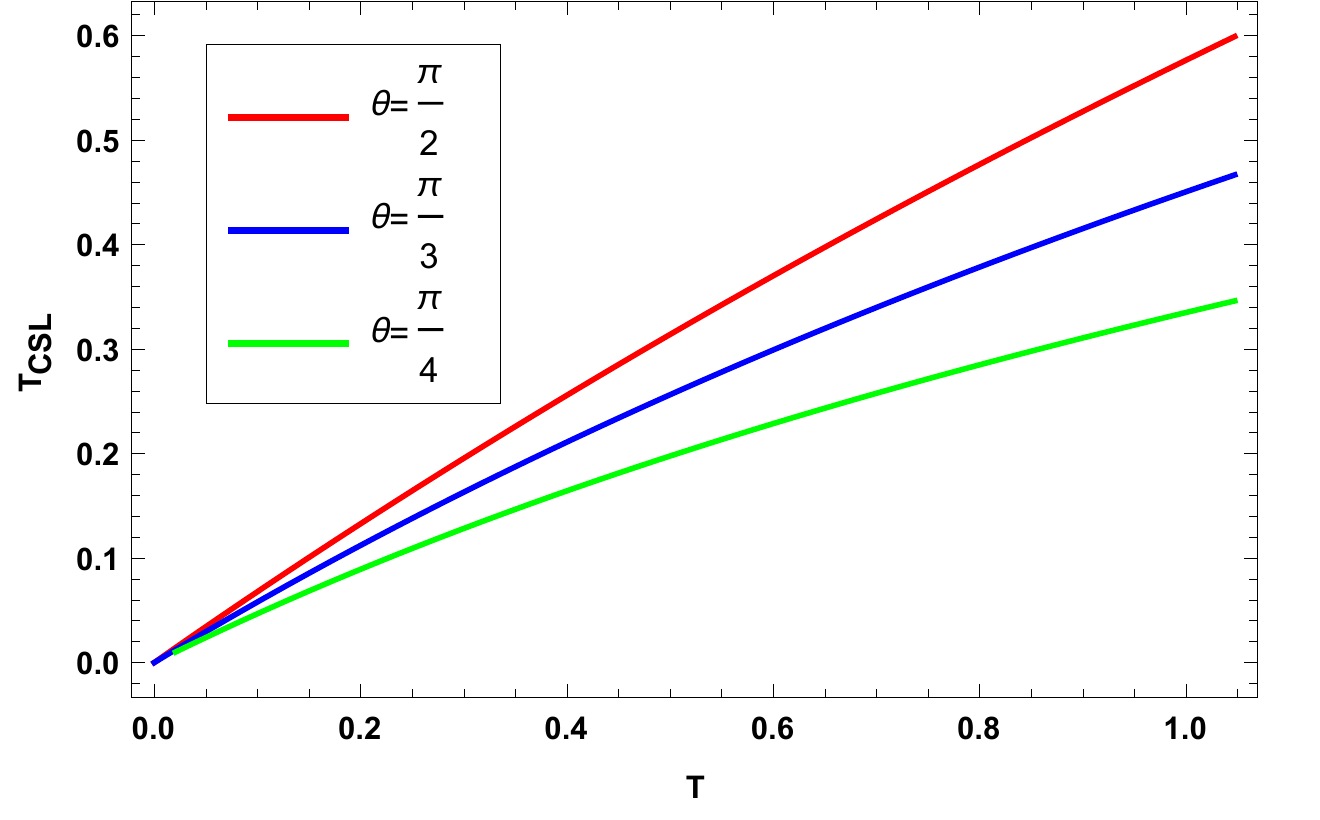}
    \caption{Here we depict $T_{CSL}$ vs $T$ and we have considered $\gamma=2$, and $\theta$ $\in$ $\{\frac{\pi}{2},\frac{\pi}{3},\frac{\pi}{4}\}$.  }
    \label{fig:TCS1}
\end{figure}

{\it Dissipative process}.--- Now, we illustrate the quantum speed limit for the coherence under dissipative process. Let us consider a two-level atom with ground state as $\op{1}$  and excited state as $\op{0}$ interacting with a dissipative environment. The corresponding jump operator is given by $L_{-} = \sqrt{\frac{\gamma}{2}} \sigma_{-}$, where  $\sigma_{-}$ is the lowering operator and $\gamma$ denotes the dissipation rate. The Lindblad master equation~\cite{Lindbland1976} governs the time evolution of atom and it is given by
\begin{equation}
\frac{ {\rm d}\rho_t}{{\rm d} t} = \frac{\gamma}{2}(\sigma_{-}\rho_t\sigma_{+}-\frac{1}{2}\{\sigma_{+}\sigma_{-},\rho_t\}),
\end{equation}
where $\sigma_{+}=|0\rangle\langle 1|$ and $\sigma_{-}=|1\rangle\langle 0|$
are the raising and the lowering operators, respectively.

If the atom the initially in a state  $\rho_{0}=\op{\psi(0)}$, where $|\psi(0)\rangle = {\cos\frac{\theta}{2}}|0\rangle +  \sin{\frac{\theta}{2}}|1\rangle$, then solution of the Lindbland equation is given by
\begin{align}\label{rho:dis}
\rho_t= &\quad e^{-\frac{\gamma t}{2} }\cos^2 {\frac{\theta}{2}}\op{0} + (1-e^{-\frac{\gamma t}{2} }\cos^2 {\frac{\theta}{2}})\op{1}\nonumber \\
&+ e^{-\frac{\gamma t}{4}} \sin {\frac{\theta}{2}} \cos {\frac{\theta}{2}} (|1\rangle\! \langle 0|+ |0\rangle\! \langle 1|).
\end{align}
To estimate bound \eqref{csl}, we require the following quantities:
\begin{align}
 C(\rho_{0})=-\sin ^2\frac{\theta }{2} \ln \left(\sin ^2\frac{\theta }{2}\right)-\cos ^2\frac{\theta }{2} \ln  \left(\cos ^2\frac{\theta }{2}\right),
 \end{align}
\begin{align}
 & C(\rho_{t})=- e^{-\frac{\gamma t}{2}}\cos ^2\frac{\theta }{2} \ln \left( e^{-\frac{\gamma  t}{2}}\cos ^2\frac{\theta }{2}\right)\nonumber\\
 &-\left(1-e^{-\frac{\gamma  t}{2}}\cos ^2\frac{\theta }{2} \right) \ln \left(1-e^{-\frac{\gamma  t}{2}}\cos ^2\frac{\theta }{2} \right)\nonumber\\
 &+\frac{1}{2} \left(\ln \left(\frac{1}{16} \left(\sqrt{2\alpha} e^{-\frac{\gamma  t}{2} }+2\right)\right)+\ln \left(2-\sqrt{2\alpha}  e^{-\frac{\gamma  t}{2}}\right)\right)\nonumber\\
 &+ \sqrt{\frac{{\alpha}}{{2}}}  e^{-\frac{\gamma  t}{2} } \tanh ^{-1}\left(e^{-\frac{\gamma  t}{2} } \sqrt{e^{\gamma  t}-4\left(e^{\frac{\gamma  t}{2}}-1\right) \cos ^4\frac{\theta }{2} }\right),
 \end{align}
 \begin{align}
\norm{\cal{L}(\rho_{t})}_{HS}^{2}= \frac{1}{32} \gamma ^2 e^{-\gamma t } \left(16 \cos ^4\frac{\theta }{2}+e^{\frac{\gamma  t}{2}}\sin ^2\theta  \right),
\end{align}
\begin{align}
 \norm{\cal{L}(\rho_{t}^{\rm D})}_{HS}^{2}= \frac{1}{2} \gamma ^2 e^{-\gamma t}\cos ^4\frac{\theta }{2},
 \end{align}
\begin{align}
&\norm{\ln{\rho^{\rm D}_{t}}}_{HS}^{2}=\left[\ln\left(e^{-\frac{\gamma  t}{2}}\cos ^2\frac{\theta }{2} \right)\right]^2\nonumber\\
&+\left[\ln\left(1- e^{-\frac{\gamma  t}{2} }\cos ^2\frac{\theta }{2}\right)\right]^2,
\end{align}
\begin{align}
&\norm{\ln{\rho_{t}}}_{HS}^{2}=\left[\ln\left(\frac{1}{4} \left(2-\sqrt{2} e^{-\frac{\gamma  t}{2} }\sqrt{\beta}\right)\right)\right]^2\nonumber\\
&\qquad+ \left[\ln\left(\frac{1}{4} \left(2+\sqrt{2} e^{-\frac{\gamma  t}{2} } \sqrt{\beta}\right)\right)\right]^2,
\end{align}
where $\alpha=3+4 \cos \theta +\cos 2 \theta -8 e^{\frac{\gamma  t}{2}}\cos ^4\frac{\theta }{2}+2 e^{\gamma  t}$ and $\beta=3-4 \left(e^{\frac{\gamma  t}{2}}-1\right)\cos\theta  -\left(e^{\frac{\gamma  t}{2}}-1\right)\cos 2\theta -3 e^{\frac{\gamma  t}{2}}+2 e^{\gamma t}$. \begin{figure}[htp]
    \centering
    \includegraphics[width=8.5cm]{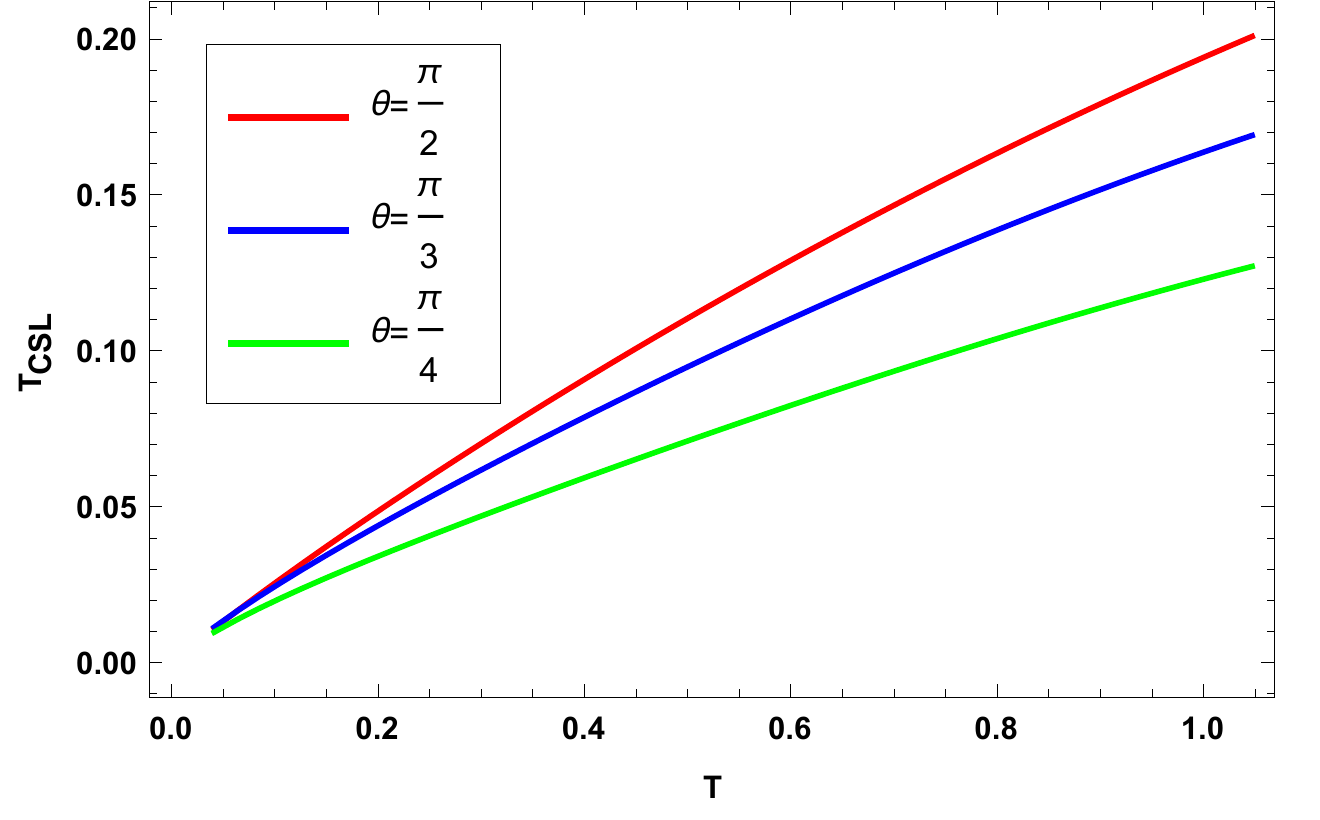}
    \caption{Here we depict $T_{CSL}$ vs $T$ and we have considered $\gamma=2$, and $\theta$ $\in$ $\{\frac{\pi}{2},\frac{\pi}{3},\frac{\pi}{4}\} $.  }
    \label{fig:TCSL2}
\end{figure}
Here, we have calculated coherence in the computational basis $\{\ket{0},\ket{1}\}$. In Fig.~\ref{fig:TCSL2}, we plot $T_{CSL}$ vs $T$ $\in$ $[0, \frac{ \pi}{3}]$ for dissipative dynamics and we have considered $\gamma=2$ and $\theta$ $\in$ $\{\frac{\pi}{2},\frac{\pi}{3},\frac{\pi}{4}\}$. 
We can observe that a maximally coherent state has a higher speed limit time of decoherence compared to other states for the dissipative process. According to Eq. \eqref{rho:dis} the dissipation time for the state $\rho_t$ is infinite. That is, the dissipative process takes an infinite amount of time to erase coherence from the quantum system. In Fig.~\ref{fig:TCS1}, we plot the bound~\eqref{csl} for finite time duration. The plot shows that the bound~\eqref{csl} is not tight as expected (see Appendix~\ref{app:sat}).

\subsection{ Bounds based on instantaneous speed}\label{AQSLI}

The method used to derive the previous bounds for quantum speed limit for information and coherence is often used in the QSL literature~\cite{Deffner2013,Cai2017}.  It may be noted that these methods can have some limitations when the evolution speed is not constant. In that case, then the notion of time averaged speed is introduced. This issue was addressed in Ref.~\cite{Connor2021} by deriving so called the action quantum speed limits. The concept of action quantum speed limit incorporates the geodesic path between the initial and final state and the notion of instantaneous evolution speed is used instead of the average evolution speed of the quantum system ~\cite{Connor2021}. Here, we opt method of speed limit of action to derive the speed limit for the von Neumann entropy and maximal information. The quantum speed limit based on instantaneous evolution speed for the entropy read as
\begin{align}
 T \geq \frac{|S(\rho_T) - S(\rho_0)|^2}{  \int_{0}^{T} (\norm{ \cal{L}_{t}({\rho_t)}}_{\rm HS} \norm{{\ln \rho_t}}_{\rm HS})^{2}{\rm d}t}.
\end{align}

The quantum speed limit based on instantaneous evolution speed for information reads as

\begin{align}
 T \geq \frac{|I(\rho_T) - I(\rho_0)|^2}{  \int_{0}^{T} (\norm{ \cal{L}_{t}({\rho_t)}}_{\rm tr} \norm{{\ln \rho_t}}_{\rm op})^{2}{\rm d}t}.
\end{align}
The detailed proof of these two bounds discussed in the Appendix~\ref{AQSL1}. 

In the quantum speed limits based on instantaneous evolution speed, the notion of time appears naturally and we do not require to introduce the notion of time as ad hoc. This is the key difference between previously obtained speed limits (bounds in ~\eqref{eq:theorem-1} and~\eqref{eq:info-t}) and speed limits based on instantaneous evolution speed. However, the speed limits based on instantaneous evolution speed has its own limitations. For example, it will be hard to achieve the bound or saturate the bound. As one can see, saturation happens when the speed is constant. This condition is very unlikely to hold for open system dynamics where typically the instantaneous speed is time-dependent.

For the sake of completeness, we have also obtained the quantum speed limit bound based on instantaneous evolution speed for coherence in the Appendix~\ref{AQSL2}.

\section{Conclusion}
Understanding of how fast one can create or erase information and coherence is important to control quantum systems for desirable information processing tasks. The limitations on the rate at which information measures change is crucial in engineering of processors to manipulate states of the quantum systems in quantum computers and communication devices. In this work, we have derived the fundamental limits on how fast the entropy, information, and quantum coherence can change for arbitrary physical processes described by completely positive maps. We obtained lower bounds on the minimal evolution time for the change in the entropy, maximal information, and quantum coherence. The quantum speed limit for information also sets a generic bound on the rate of information production. We showcase an application of our main result by answering the question on the minimum time required to erase the information of a given quantum state, where erasure could be under a general resetting operation or via thermalization process.
The quantum speed limit on coherence for dephasing process answers the question how fast system can undergo decoherence. We have illustrated the quantum speed limit for coherence for pure dephasing as well as dissipative processes. 
 
Unlike the quantum speed limit for state evolution under unitary dynamics, an important aspect of the quantum speed limits for information given in ~\eqref{eq:info-t} and~\eqref{csl} is that, in general, they cannot be made tight (see Appendix~\ref{app:sat}). If we demand that the inequality is saturated, then the time evolved density operator may not remain a positive operator for general open system dynamics. This is also evident from the examples of non-unitary quantum dynamics we have considered in this paper. We also provide a condition when the inequality may be saturated. As a future work, it would be interesting to find speed limits on informational measures that are tighter than the current bounds.

All the bounds derived in our paper (Theorems \ref{thm:1}--\ref{thm:3}) make use of the geometric approach,
where the lower bounds to entropy, information and coherence depend on the time averaged notion of
the root mean square evolution speed. One may ask if it possible to obtain quantum speed limits 
where the bounds do not depend on the
average speed. In this respect, the action speed limits have been proved where a different derivation based on
the notion of action allowed to incorporate the notion of instantaneous
evolution speed into the QSL bound. Towards the end of our paper, we have also presented similar bounds for entropy, maximal information, and coherence. In future, we will explore more on the physical applications of the geometric QSL and instantaneous speed based QSL.
The speed limits, presented in this work, may find application in the study of quantum thermodynamics, open quantum systems, quantum control theory, and engineering quantum technologies.

\medskip

\begin{acknowledgments}
The authors thank  Lucas C.~C\'eleri for feedback. BM acknowledges the support of the INFOSYS scholarship. This work is supported by the J. C. Bose Fellowship from the Department of Science and Technology (DST)
India under Grant No.~JCB/2018/000038 (2019–2024). SD acknowledges Individual Fellowships at Universit\'e libre de Bruxelles; this project receives funding from the European Union’s Horizon 2020 research and innovation program under the Marie Sk\l{}odowska-Curie Grant Agreement No.~801505. SD is
thankful to the Harish-Chandra Research Institute, Prayagraj, India for the hospitality during his visit where part of this work was done.
\end{acknowledgments}

\newpage

\appendix
\section{Rate of coherence change}\label{app:co-change}
We consider the density operator $\rho^{{\rm D}}_{t}$ that is incoherent state in the reference basis $\{|i\rangle\}_{i=0}^{d-1}$, which is fixed and independent of time, and formed after the action of completely dephasing channel on the density operator $\rho_{t}$ :
\begin{equation}\label{eq:rho-d}
    \rho_{t}^{\rm D}= \sum_{i}\bra{i}\!\rho_{t}\!\ket{i}\op{i}.
\end{equation}
We now present a Lemma that we use in Section~\ref{sec:coherence-change}.
\begin{lemma}\label{lem:co-change}
 The rate of change of coherence $C(\rho_t):= S(\rho^{\rm D}_{t})-S(\rho_{t})$ is given by
\begin{equation}\label{eq:coherence-rate-change}
   \frac{{\rm d} }{{\rm d}t}C(\rho_t)= -{\rm tr}\{\mathcal{L}_{t}({\rho^{\rm D}_t}) \ln \rho^{\rm D}_t\}+{\rm tr}\{\mathcal{L}_{t}({\rho_t}) \ln \rho_t\}.
\end{equation}
\end{lemma}
\begin{proof}
The rate of change of entropy of the completely dephased density operator $\rho^{\rm D}_t$ is given as:
\begin{align}
& \diff{S(\rho^{\rm D}_t)}{t} \nonumber\\
= & -\tr\left[ \diff{\rho^{\rm D}_t}{t}\ln \rho^{\rm D}_t\right]\\
= &- \tr \left[\sum_{i,j}\bra{i}{\mathcal{L}_t(\rho_t)}\ket{i}\op{i} \ln \bra{j}\rho_t\ket{j}\op{j}  \right]\\
= & -\tr \left[\sum_{i,j}\bra{i}{\mathcal{L}_t(\rho_t)}\ket{i} \ln \bra{j}\rho_t\ket{j} \delta_{i,j}\ket{i}\!\bra{j}  \right]\\
= &-\sum_{i}\bra{i}{\mathcal{L}_t(\rho_t)}\ket{i}\ln \bra{i}\rho_t\ket{i}\label{eq:diag-rate} .
\end{align}
From the identity~\eqref{eq:rho-d}, we have,
\begin{align}
     \mathcal{L}_{t}(\rho^{\rm D}_t) = \diff{\rho^{\rm D}_t}{t} = \sum_{i}\bra{i}{\mathcal{L}_t(\rho_t)}\ket{i}\op{i}.\label{eq:l-rho-t}
\end{align}
Invoking the identity~\eqref{eq:l-rho-t} in the identity~\eqref{eq:diag-rate} and using identity of the rate of entropy change in Lemma~\ref{lem:ent-change}, we get the desired identity~\eqref{eq:coherence-rate-change}.
\end{proof}

\section{Proof of Theorem \ref{thm:2}}\label{QSLISL}

The information content $I(\rho_{t})$ of the time evolved density operator $\rho_{t}$ is given by
\begin{equation}
     I(\rho_t)= { \ln} (d)-S(\rho_{t}).
\end{equation}
After differentiating the above equation with respect to time $t$, we obtain 
\begin{equation}
   \frac{ {\rm d}}{{\rm d}t}I(\rho_t)= {\rm tr}\{\cal{L}_{t}({\rho_t})  \ln \rho_t\}.
\end{equation}

Let us take the absolute value of above equation and apply the H\"older's inequality, then we obtain the following inequality
\begin{equation}
   \left|\frac{{\rm d} }{{\rm d}t}I(\rho_t)\right| \leq \norm{\cal{L}_{t}({\rho_t})}_{\rm tr} \norm{{ \ln \rho_t}}_{\rm op}.
\end{equation}

After integrating above equation with respect to time $t$, we obtain 
\begin{equation}
  \int_{0}^{T}{\rm d}t \left|\frac{{\rm d} }{{\rm d}t}I(\rho_t)\right| \leq \int_{0}^{T}  \norm{\cal{L}_{t}({\rho_t})}_{\rm tr} \norm{{ \ln \rho_t}}_{\rm op} {\rm d}t.
\end{equation}
We apply the Cauchy--Schwarz inequality on the right hand side of the above inequality, we get  
\begin{align}
 &\frac{1}{T} \int_{0}^{T}{\rm d}t \left|\frac{{\rm d} I(\rho_t)}{{\rm d}t}\right|\nonumber\\ & \leq  \sqrt{\frac{1}{T}\int_{0}^{T} \norm{ \cal{L}_{t}({\rho_t)}}^{2}_{\rm tr}{\rm d}t} \sqrt{\frac{1}{T}\int_{0}^{T}  \norm{{ \ln \rho_t}}^{2}_{\rm op} {\rm d}t}\label{eqs}.
\end{align}
From the inequality~\eqref{eqs} we can obtain the following bound
\begin{equation}
  T\geq  \frac{|I(\rho_T) - I(\rho_0)|}{ {\Lambda}^{\rm rms}_{T}\overline{\norm{{\ln \rho_t}}^{2}_{\rm op}}}.
\end{equation}

\section{On the saturation of speed limits}\label{app:sat}
For any two normal operators $A$ and $B$, the following equality holds if and only if $A$ and $B$ are linearly dependent, i.e., $A=c B$ where $c\in\mathbbm{C}$:
\begin{equation}
     |{\rm tr}(AB)|=\sqrt{{\rm tr}(A^{\dagger}A){\rm tr}(B^{\dagger}B)}.
\end{equation}
For the saturation of the quantum speed limits for the information which are derived invoking the Cauchy-Schwarz inequality (e.g, \eqref{eq:ent-bound}), a necessary criterion is that the involved density operator $\rho_t$ and its Liouvillian $\mathcal{L}_t(\rho_t)
= {\dot {\rho}_t}$ should be such that the Cauchy-Schwarz inequality itself is saturated.

Note that any density operator $\rho_t$ can be expressed as 
$$\rho_t=\cfrac{\exp{(-K(t))}}{\tr{[\exp{(-K(t))}}]}$$ for some Hermitian operator $K(t)$, which gives $- \ln \rho_t=K(t)+\tr{[\exp{(-K(t))}]}\mathbbm{1}$. For the Cauchy-Schwarz inequality to be saturated in \eqref{eq:ent-bound}, we need to have for some $a\in\mathbbm{R}$
\begin{equation}
\mathcal{L}_t(\rho_t)= \dot{\rho_t}=a (K(t)+\tr{[\exp{(-K(t))}]}\mathbbm{1}),
\end{equation}
which implies 
\begin{equation}
    \rho_T=\rho_0+a\int_0^T(K(t)+\tr{[\exp{(-K(t))}]}\mathbbm{1})\mathrm{d}t.
\end{equation}
The above identity in general is not true as adding a Hermitian operator to a positive operator may not keep the positivity. 
In other words, if we demand that the inequality is saturated, then the time evolved density operator may not remain a positive operator for general open system dynamics. Thus, the equality condition for quantum speed limits for information cannot be reached in general.
Though, this may hold only in special cases as the second term on the right hand side has to be a positive operator and needs to be traceless for the identity to hold. Our proof also gives a condition when the equality may hold. For example, if the open system dynamics not only keeps $\rho_t$ positive, but also keeps $-\ln \rho_t$ positive through out the evolution, then equality may be achieved. 
This is a highly non-trivial task and we leave this as a future problem.

\section{Derivation of quantum speed limit based on instantaneous evolution speed for information}\label{AQSL1}

The  Cauchy--Schwarz inequality for any two continuous functions states that
\begin{equation*}
    \int_{0}^{T}f(t) ~~ g(t) {\rm d}t \leq \sqrt{\int_{0}^{T} f(t)^{2}{\rm d}t} \sqrt{\int_{0}^{T} g(t)^{2}{\rm d}t},
\end{equation*}
where $f(t)$ and $g(t)$ are real functions that are continuous on the closed interval $[0,T]$.

We apply the Cauchy--Schwarz inequality on the right hand side of the inequality \eqref{MER}, by setting $g = 1$, we obtain  
\begin{align}
 & \int_{0}^{T}{\rm d}t \left|\frac{{\rm d} S(\rho_t)}{{\rm d}t}\right|\nonumber\\ & \leq  \sqrt{\int_{0}^{T} (\norm{ \cal{L}_{t}({\rho_t)}}_{\rm HS} \norm{{\ln \rho_t}}_{\rm HS})^{2}{\rm d}t} \sqrt{\int_{0}^{T}  {\rm d}t}\label{eqs-2}.
\end{align}

After integrating above equation with respect to time $t$, we obtain 
\begin{align}
 & |S(\rho_T) - S(\rho_0)| \leq \sqrt{T} \sqrt{\int_{0}^{T} (\norm{ \cal{L}_{t}({\rho_t)}}_{\rm HS} \norm{{\ln \rho_t}}_{\rm HS})^{2}{\rm d}t} \label{eqs-3}.
\end{align}

The above inequality can be re-written in the following form
\begin{align}
 T \geq \frac{|S(\rho_T) - S(\rho_0)|^2}{  \int_{0}^{T} (\norm{ \cal{L}_{t}({\rho_t)}}_{\rm HS} \norm{{\ln \rho_t}}_{\rm HS})^{2}{\rm d}t}.
\end{align}

Similarly, for maximal information, one can obtain the following bound
\begin{align}
 T \geq \frac{|I(\rho_T) - I(\rho_0)|^2}{  \int_{0}^{T} (\norm{ \cal{L}_{t}({\rho_t)}}_{\rm tr} \norm{{\ln \rho_t}}_{\rm op})^{2}{\rm d}t}.
\end{align}

Hence, the proofs.

\section{Derivation of quantum speed limit based on instantaneous evolution speed for coherence}\label{AQSL2}

Let us apply the Cauchy--Schwarz inequality on the right hand side of the inequality \eqref{MCR}, by setting $g(t) = 1$, we obtain  
\begin{align}
  \int_{0}^{T}{\rm d}t\left|\frac{{\rm d}}{{\rm d}t} C(\rho_t)\right| \leq& \sqrt{\int_{0}^{T}  \norm{\mathcal{L}_{t}({\rho^{\rm D}_t})}^{2}_{\rm HS}\norm{\ln \rho^{\rm D}_t}^{2}_{\rm HS}{\rm d}t} \sqrt{\int_{0}^{T}   {\rm d}t}\nonumber\\
  &+\sqrt{\int_{0}^{T}  \norm{\mathcal{L}_{t}({\rho_t})}^{2}_{\rm HS}\norm{{\ln \rho_t}}^{2}_{\rm HS}{\rm d}t} \sqrt{\int_{0}^{T}   {\rm d}t}.
\end{align}

After integrating above equation with respect to time $t$, we obtain 
\begin{align}
  | C(\rho_T)-C(\rho_0)| \leq&\sqrt{T} [\sqrt{\int_{0}^{T}  \norm{\mathcal{L}_{t}({\rho^{\rm D}_t})}^{2}_{\rm HS}\norm{\ln \rho^{\rm D}_t}^{2}_{\rm HS}{\rm d}t} \nonumber\\
  &+\sqrt{\int_{0}^{T}  \norm{\mathcal{L}_{t}({\rho_t})}^{2}_{\rm HS}\norm{{\ln \rho_t}}^{2}_{\rm HS}{\rm d}t}].
\end{align}

The above inequality can be re-written in the following form
\begin{equation}
  T\geq  \frac{|C(\rho_T) - C(\rho_0)|^2}{ [\sqrt{\Lambda_{T}^{D}}
  +\sqrt{\Lambda_{T}}]^2 },
\end{equation}

where $\Lambda_{T}^{D}=\int_{0}^{T}  \norm{\mathcal{L}_{t}({\rho^{\rm D}_t})}^{2}_{\rm HS}\norm{\ln \rho^{\rm D}_t}^{2}_{\rm HS}{\rm d}t$ and  $\Lambda_{T}=\int_{0}^{T}  \norm{\mathcal{L}_{t}({\rho_t})}^{2}_{\rm HS}\norm{{\ln \rho_t}}^{2}_{\rm HS}{\rm d}t$.

\bibliography{main}

 \end{document}